\documentclass[a4paper]{article}

\usepackage{amsfonts,amsmath,amssymb,amsthm}
\usepackage{latexsym}
\usepackage{mathrsfs}
\usepackage{graphicx}
\usepackage{indentfirst}
\usepackage{color}
\usepackage{todonotes}
\usepackage{setspace}
\usepackage{mathdots}
\usepackage{a4wide}

\usepackage{authblk}
\usepackage[all]{xy}

\usepackage[ruled,boxed]{algorithm2e}
\SetKwInOut{Input}{input}\SetKwInOut{Output}{output}
\SetKwBlock{Begin}{}{}

\usepackage[final]{showkeys}
\definecolor{refkey}{gray}{.05}
\definecolor{labelkey}{gray}{.75}

\DeclareMathAlphabet{\mathpzc}{OT1}{pzc}{m}{it}

\newtheorem{theorem}{Theorem}
\newtheorem{corollary}[theorem]{Corollary}

\newtheorem{lemma}[theorem]{Lemma}
\newtheorem{proposition}[theorem]{Proposition}

\theoremstyle{definition}
\newtheorem{definition}[theorem]{Definition}
\newtheorem{remark}[theorem]{Remark}

\newcommand{\F}{\mathbb{F}}

\newcommand{\N}{\mathbb{N}}
\newcommand{\R}{\mathbb{R}}

\newcommand{\PP}{\mathcal{P}}
\newcommand{\II}{\mathcal{I}}

\newcommand{\Gr}[3]{\mathfrak{G}_{#1}(#2,#3)}

\newcommand{\Cvs}{\mathcal{C}}
\newcommand{\Svs}{\mathsf{S}}
\newcommand{\Rvs}{\mathcal{R}}

\newcommand{\Vvs}{\mathcal{V}}
\newcommand{\Uvs}{\mathcal{U}}

\newcommand{\mat}[1]{\left(\begin{matrix}#1\end{matrix} \right)}
\newcommand{\rs}{\mathrm{rowsp}}

\newcommand{\rank}{\mathrm{rank}}
\newcommand{\Min}[3]{[{#1}, {#2}]_{#3}}

\newcommand{\subm}[3]{\big(#1;#2\big)_{#3}}

\newcommand{\Al}{\alpha}

\newcommand{\Gl}{\gamma}
\newcommand{\Ll}{\lambda}

\title{An Algebraic Approach for Decoding Spread Codes}
\author[1]{E. Gorla \thanks{The author was supported by the Swiss
    National Science Foundation under grant no. 123393.}}

\author[2]{F. Manganiello \thanks{The author was partially supported by
    the Swiss National Science Foundation under grants no.\ 126948 and
    no.\ 135934.}}

\author[3]{J. Rosenthal \thanks{The author was partially supported by
    the Swiss National Science Foundation under grant no.\ 126948.}}

\affil[1]{Mathematics Institute\\University of Basel}
\affil[2]{Department of Electrical and Computer Engineering\\
  University of Toronto}
\affil[3]{Mathematics Institute\\University of Z\"urich} 

\date{}

\begin{document}
\maketitle







\begin{abstract}
  In this paper we study spread codes: a family of constant-dimension
  codes for random linear network coding. In other words, the
  codewords are full-rank matrices of size $k\times n$ with entries in
  a finite field $\F_q$. Spread codes are a family of optimal codes with
  maximal minimum distance. We give a minimum-distance decoding
  algorithm which requires $\mathcal{O}((n-k)k^3)$ operations over an extension
  field $\F_{q^k}$. Our algorithm is more efficient than the previous
  ones in the literature, when the dimension $k$ of the codewords is
  small with respect to $n$. The decoding algorithm takes advantage of
  the algebraic structure of the code, and it uses original results on
  minors of a matrix and on the factorization of polynomials over
  finite fields.
\end{abstract}

\section{Introduction}

Network coding is a branch of coding theory that arose in 2000 in the
work by Ahlswede, Cai, Li and Yeung \cite{ah00}. While classical
coding theory focuses on point-to-point communication, network coding
focuses on multicast communication, i.e., a source communicating with
a set of sinks. The source transmits messages to the sinks over a
network, which is modeled as a directed multigraph. Some examples of
multicast communication come from Internet protocol applications of
streaming media, digital television, and peer-to-peer networking.
 
The goal in multicast communication is achieving maximal information
rate. Informally, this corresponds to maximizing the amount of
messages per transmission, i.e., per single use of the network. Li,
Cai and Yeung in \cite{li03} prove that maximal information rate can
be achieved in multicast communication using linear network coding,
provided that the size of the base field is large enough.

The algebraic aspect of network coding emerged with the work by
K\"otter and Kschischang \cite{ko08}. The authors introduced a new
setting for random linear network coding: Given the linearity of the
combinations, the authors suggest to employ subspaces of a given
vector space as codewords. Indeed, subspaces are invariant under
taking linear combinations of their elements. Let $\PP(\F_q^n)$ be the
set of all $\F_q$-linear subspaces of $\F_q^n$. They show that
$\PP(\F_q^n)$ is a metric space, with distance
\[d(\Uvs,\Vvs)=\dim(\Uvs+\Vvs)-\dim(\Uvs\cap \Vvs) \mbox{ for all }
\Uvs,\Vvs\in \PP(\F_q^n).\] K\"otter and Kschischang define network
codes to be subsets of $\PP(\F_q^n)$. In particular, they define
constant-dimension codes as subsets, whose elements have all the
same dimension. Notions of errors and erasures compatible with the
new transmission model are introduced in~\cite{ko08}. In addition,
upper and lower bounds for the cardinality of network codes are
established in~\cite{ko08,et08p}.
\medskip 

We review here some of the constructions of constant-dimension codes
present in the literature. The first one is introduced by
K\"otter and Kschischang in~\cite{ko08}. The construction uses
evaluation of linearized polynomials over a subspace. The codes that
one obtains are called Reed-Solomon-like codes, because of the
similarities with Reed-Solomon codes in classical coding theory. Due
to their connection with the rank-metric codes introduced
in~\cite{ga85a}, these codes are also called lifted rank-metric
codes. K\"otter and Kschischang devise a list-$1$
minimum-distance decoding algorithm for their codes. Spread codes,
which are the subject of this paper, were first introduced by the
authors in~\cite{ma08p}. Spread codes contain the codes with maximal
minimum distance in~\cite{ko08}. Another family
of network codes, based on $q$-analogs of designs, appears
in~\cite{ko08p}. Aided by computer search, the authors find
constant-dimension codes based on designs with big cardinality.
Another family of codes is constructed in~\cite{et08u}. The
construction builds on that of Reed-Solomon-like codes, and the
codes that the authors obtain contain them. The construction is also
based on binary constant-weight codes, Ferrer diagrams, and
rank-metric codes. The proposed decoding algorithm operates on two 
levels: First one decodes a constant-weight code, then one applies a
decoding algorithm for rank-metric codes. In~\cite{sk10} Skachek
introduces a family of codes, that is a sub-family of the one
in~\cite{et08u}. In~\cite{ma10p} the authors introduce another family
of codes, which they obtain by evaluating pairs of linearized
polynomials. The codes obtained can be decoded via a list decoding
algorithm, which is introduced in the same work.
\medskip 

This work focuses on \emph{spread codes} which are a family of
constant-dimension codes first introduced in \cite{ma08p}. Spreads of
$\F_q^n$ are a collection of subspaces of $\F_q^n$, all of the same
dimension, which partition the ambient space. Such a family of
subspaces of $\F_q^n$ exists if and only if the dimension of the
subspaces divides $n$. The construction of spread codes is based
on the $\F_q$-algebra $\F_q[P]$ where $P\in GL_k(\F_q)$ is the
companion matrix of a monic irreducible polynomial of degree
$k$. Concretely, we define spread codes as 
\[
\Svs_r=\left\{\rs\mat{A_1 & \cdots & A_r} \in \Gr{\F_q}{k}{n} \mid
  A_i\in \F_q[P]\ \mbox{ for all } i\in \{1,\dots,r\} \right\}
\]
where $\Gr{\F_q}{k}{n}$ is the Grassmannian of all subspaces of
$\F_q^n$ of dimension $k$.

Since spreads partition the ambient space, spread codes are
optimal. More precisely, they have maximum possible minimum distance
$2k$, and the largest possible number of codewords for a code with
minimum distance $2k$. Indeed, they achieve the anticode bound
from~\cite{et08p}. This family is closely related to the family of
Reed-Solomon-like codes introduced in~\cite{ko08}. We discuss the
relation in detail in Section~\ref{ss:RSL}. In Lemma~\ref{n=rk_RSL},
we show how to extend to spread codes the existing decoding algorithms
for Reed-Solomon-like codes and rank-metric codes.

The structure of the spreads that we use in our construction helps us
devise a minimum-distance decoding algorithm, which can correct up to
half the minimum distance of $\Svs_r$. In Lemma~\ref{l:n=rk} we reduce
the decoding algorithm for a spread code $\Svs_r$ to at most $r-1$
instances of the de\-co\-ding algorithm for the special case
$r=2$. Therefore, we focus 
on the design of a decoding algorithm for the spread code 
\[\Svs=\Svs_2=\left\{ \rs\mat{A_1 & A_2}\in \Gr{\F_q}{k}{2k}\mid
  A_1,A_2\in\F_{q}[P]\right\}.\]
\medskip

The paper is structured as follows. In Section~\ref{s:spread_cons} we
give the construction of spread codes, discuss their main
properties. In Subsection~\ref{ss:notation} we introduce the main
notations. In Subsection~\ref{ss:RSL} we discuss the relation between
spread codes and Reed-Solomon-like codes, which is given explicitly in
Proposition~\ref{p:spread_lift}. Proposition~\ref{p:compl} shows how
to apply a minimum-distance decoding algorithm for Reed-Solomon-like
codes to spread codes, and estimates the complexity of decoding a
spread code using such an algorithm.

The main results of the paper are contained in
Section~\ref{s:dec_alg}. In Subsection~\ref{ss:mat} we prove some 
results on matrices, which will be needed for our decoding
algorithm. Our main result is a new minimum-distance decoding 
algorithm for spread codes, which is given in pseudocode as
Algorithm~\ref{a:dec_alg}. The decoding algorithm is based on
Theorem~\ref{t:suitable_poly}, where we explicitly construct the
output of the decoder. Our algorithm can be made more efficient when
the first $k$ columns of the received word are linearly
independent. Proposition~\ref{p:unique_non_sing} and
Corollary~\ref{c:non_sing} contain the theoretical results behind this
simplification, and the algorithm in pseudocode is given in
Algorithm~\ref{a:non_sing}.
Finally, in Section~\ref{sss:complexity} we
compute the complexity of our algorithm. Using the results from
Subsection~\ref{ss:RSL}, we compare it with the complexity of the
algorithms in the literature. It turns out that our algorithm is more
efficient than the all the known ones, provided that $k\ll n$.


\section{Preliminaries and notations}\label{s:spread_cons}


\begin{definition}[{\cite[Section~4.1]{hi98}}]
  A subset $S\subset \Gr{\F_q}{k}{n}$ is a spread if it satisfies
  \begin{itemize}
  \item $\mathcal{U}\cap \mathcal{V} = \{0\}$ for all
    $\mathcal{U},\mathcal{V}\in S$ distinct, and
  \item $\F_q^n=\bigcup_{\mathcal{U}\in S}\mathcal{U}.$
  \end{itemize}
\end{definition}


\begin{theorem}[{\cite[Theorem~4.1]{hi98}}]
  A spread exists if and only if $k\mid n$.
\end{theorem}

In~\cite{ma08p} we give a construction of spreads suitable for use in
Random Linear Network Coding (RLNC). Our construction is based on
companion matrices.

\begin{definition}
  Let $\F_q$ be a finite field and $p=\sum_{i=0}^k p_ix^i\in \F_q[x]$
  a monic polynomial. The companion matrix of $p$
  \[P=\mat{0 & 1 & 0 & \cdots & 0 \\ 0 & 0 & 1 & & 0 \\
    \vdots & & & \ddots & \vdots \\ 0 & 0 & 0 &  & 1\\
    -p_0 & -p_1 & -p_2 & \cdots & -p_{k-1} }\in \F_q^{k\times k}.\]
\end{definition}
\medskip
Let $n=rk$ with $r>1$, $p\in \F_q[x]$ a monic irreducible polynomial
of degree $k$ and $P\in \F_q^{k\times k}$ its companion
matrix.

\begin{lemma}
  The $\F_q$-algebra $\F_q[P]$ is a finite field, i.e., $\F_q[P]\cong
  \F_{q^k}$.
\end{lemma}

This is a well-known fact (see \cite[page 64]{li94}).

\begin{lemma}
  Let $\varphi:\F_{q^k}\rightarrow \F_q[P]$ be a ring
  isomorphism. Denote by
  \[\mathbb{P}^{r-1}(\F_{q^k}):=(\F_{q^k}^r\setminus \{0\})/\sim\] the
  projective space, where $\sim$ is the following equivalence
  relation
  \[v\sim w \iff \exists \lambda\in \F_{q^k}^* \mbox{ such that }
  v=\lambda w,\] where $v,w\in \F_{q^k}^r\setminus \{0\}$. Then the
  map
  \[
  \begin{array}{cccc}
    \tilde\varphi: & \mathbb{P}^{r-1}(\F_{q^k}) & \rightarrow &
    \Gr{\F_q}{k}{n} \\
    & [v_1:\dots:v_r] & \mapsto & \rs \mat{\varphi(v_1)& \cdots &
      \varphi(v_r)}.
  \end{array}
  \]
is injective.
\end{lemma}

\begin{proof}
  Let $v=[v_1:\dots:v_r],w=[w_1:\dots:w_r]\in
  \mathbb{P}^{r-1}(\F_{q^k})$. If $\tilde\varphi(v)=\tilde\varphi(w)$
  there exists an $M\in GL_k(\F_q)$ such that 
  \begin{eqnarray}
    \mat{\varphi(v_1)& \cdots & \varphi(v_r)} &=&
    M\mat{\varphi(w_1)& \cdots & \varphi(w_r)} \nonumber\\ 
      & =& \mat{M\varphi(w_1)& \cdots &
      M\varphi(w_r)} \label{e:proj}
  \end{eqnarray} 
  Let $i_v,i_w\in \{1,\dots,r\}$ be the least indices such that
  $\varphi(v_{i_v})\neq 0$ and $\varphi(w_{i_w})\neq 0$. From \eqref{e:proj} it
  follows that $i_v=i_w$. Since, without loss of generality, we can
  consider $v_{i_v}=w_{i_w}=1$, it follows that
  $\varphi(v_{i_v})=\varphi(w_{i_w})=I$ and consequently
  $M=I$. Then, \eqref{e:proj} becomes  
  \[\mat{\varphi(v_1)& \cdots & \varphi(v_r)}=
  \mat{\varphi(w_1)& \cdots & \varphi(w_r)}\]
  leading to $v=w$.
\end{proof}

\begin{theorem}[{\cite[Theorem~1]{ma08p}}]\label{t:def_spread}
$\Svs_r:=\tilde\varphi(\mathbb{P}^{r-1}(\F_{q^k}))$ is a spread of
$\F_q^n$ for $n=rk$. 
\end{theorem}

\begin{definition}[{\cite[Definition~2]{ma08p}}]
  We call \emph{spread codes} of $\Gr{\F_q}{k}{n}$ the subsets
  $\Svs_r\subset\Gr{\F_q}{k}{n}$ from Theorem \ref{t:def_spread}.
\end{definition}

\begin{remark}
Notice that
  \begin{eqnarray*}
    \Svs_r=\left\{\rs\mat{A_1 & \cdots & A_r} \in \Gr{\F_q}{k}{n} \mid
      A_i\in \F_q[P]\ \mbox{ for all }
      i\in \{1,\dots,r\} \right\} 
  \end{eqnarray*}
In order to have a unique representative for the elements of $\Svs_r$,
we bring the matrices $\mat{A_1 & \cdots & A_r}$ in row reduced
echelon form.
\end{remark}

\begin{lemma}[{\cite[Theorem 1]{ma08p}}]
  Let $\Svs_r\subset\Gr{\F_q}{k}{n}$ be a spread code. Then
  \begin{enumerate}
  \item $d(\mathcal{U},\mathcal{V})=d_{\min}(\Svs_r)=2k$, for all
    $\mathcal{U},\mathcal{V}\in \Svs_n$ distinct, i.e., the code has
    maximal minimum distance, and
  \item $|\Svs_r|=\frac{q^n-1}{q^k-1}$, i.e., the code has maximal
    cardinality with respect to the given minimum distance.
  \end{enumerate}
\end{lemma}

\begin{remark}
In~\cite{tr10p} the authors show that spread codes are an example of
\emph{orbit codes}. Moreover, in~\cite{tr11} it is shown that some
spread codes are cyclic orbit codes under the action of the cyclic
group generated by the companion matrix of a primitive polynomial. 
\end{remark}

\begin{definition}
A vector space $\Rvs\in \Gr{\F_q}{\tilde{k}}{rk}$ is \emph{uniquely
  decodable} by the spread code $\Svs_r\subset\Gr{\F_q}{k}{n}$ if
\begin{eqnarray*}
\mbox{there exists a } \Cvs\in \Svs_r \mbox{ such that } d(\Rvs,\Cvs)<
\frac{d_{\min}(\Svs_r)}2=k.
\end{eqnarray*} 
\end{definition}

In Section \ref{s:dec_alg} we devise a minimum-distance decoding algorithm
for uniquely decodable received spaces. 

\subsection{Further notations}\label{ss:notation}

We introduce in this subsection the notation we use in the paper.

\begin{definition}
  Let $s\in \N$ with $s<k$ and denote by
  $L_{\F_{q^n}}^s\subset \F_{q^n}[x]$ the set of linearized
  polynomials of degree less than $q^s$. Equivalently, $f\in
  L_{\F_{q^n}}^s$  if and only if $f=\sum_{i=0}^{s-1}f_ix^{q^i}$ for
  some $f_i\in\F_{q^n}$.
\end{definition}

In the rest of the work we denote $q$-th power exponents such as
$x^{q^i}$ with $x^{[i]}$. 

\medskip

Let $\F_q$ be a finite field with $q$ elements, and
  let $p\in \F_q[x]$ be a monic irreducible polynomial of degree
  $k>1$. $P\in GL_k(\F_q)$ denotes the companion matrix of $p$, and
  $S\in GL_k(\F_{q^k})$ is a matrix which diagonalizes $P$.

\medskip

We denote by $\Delta(x):=\mathrm{diag}(x,x^{[1]},\dots,x^{[k-1]})\in
  \F_q[x]^{k\times k}$ a diagonal matrix, whose entry in position
  $(i+1,i+1)$ is $x^{[i]}$ for $i=0,\ldots,k-1$.

\medskip

  Let $M$ be a matrix of size $k\times k$ and let $J=(j_1,\dots,j_s)$,
  $L=(l_1,\dots,l_s)\in \{1,\dots,k\}^s$. $[J;L]_M$ denotes the
  minor of the matrix $M$ corresponding to the submatrix
  $\subm{J}{L}{M}$ with row indices $j_1,\dots,j_s$ and column indices
  $l_1,\dots,l_s$. We skip the suffix $M$ when the matrix is clear
  from the context.

\medskip

We introduce some operations on tuples. Let $K=(i_1,\dots,i_s)\in \{1,\dots,k\}^s$. 
\begin{itemize}
\item $i\in K$ means that $i\in \{i_1,\dots,i_s\}$.
\item $L\subset K$ means that $L=(i_{l_1},\dots, i_{l_k})$ for $1\leq
  l_1<\dots<l_k\leq s$.
\item $|K|:=s$ is the length of the tuple.
\item $K\cap J$ denotes the $L\subset K,J$ such that $|L|$ is maximal.
\item If $J=(j_1,\dots,j_r)$ then $I\cup
  J:=(i_1,\dots,i_s,j_1,\dots,j_r)$, i.e., $\cup$ denotes the
  concatenation of tuples.
\item If $J\subset K$ then $K\setminus J$ denotes the $L\subset K$
  with $|L|$ maximal such that $J\cap L =\emptyset$ where
  $\emptyset$ is the empty tuple.
\item $\min K=\min\{i\mid i\in K\}$, with the convention that $\min
  \emptyset > \min K$ for any $K$.
\end{itemize}

\medskip

We define the \emph{non diagonal rank} of a matrix as follows.

\begin{definition}
  Let $N\in \F_{q}^{k\times k}$. We define the \emph{non diagonal rank} of
  $N$ as 
  \[\mathrm{ndrank}(N):=\min\{t\in \N \mid [J,L]_N=0 \ \mbox{ for all } J,L\in
  \{1,\dots,k\}^t, \ J\cap L=\emptyset\}-1.\]
\end{definition}

\medskip

At last, algorithms' complexities are expressed as $\mathcal{O}(\F;p(n,k))$,
which corresponds to performing $\mathcal{O}(p(n,k))$ operations over
a field $\F$, where $n,k$ are given parameters.

\subsection{Relation with Reed-Solomon-like codes}\label{ss:RSL}

Reed-Solomon-like codes, also called lifted rank-metric codes, are a
class of constant-dimension codes introduced in \cite{ko08}. They are
strictly related to maximal rank
distance codes as introduced in \cite{ga85a}. We give here an
equivalent definition of these codes. 

\begin{definition}
  Let $\F_q\subset \F_{q^n}$ be finite fields. Fix some
  $\F_q$-linearly independent elements $\Al_1,\dots,\Al_k\in
  \F_{q^n}$.  Let $\psi:\F_{q^n}\rightarrow \F_q^n$ be an
  isomorphism of $\F_q$-vector spaces. A \emph{Reed-Solomon-like (RSL) code}
  is defined as
  \[RSL_{\F_{q^n}}^s:=\left\{ \rs \left. \left(I\ 
      \begin{array}{c}
        \psi(f(\Al_1)) \\ \vdots \\ \psi(f(\Al_k))
      \end{array}
\right)\; \right| \; f\in L_{\F_{q^n}}^s
\right\}\subseteq \Gr{\F_q}{k}{k+n}.\]
\end{definition}

The following proposition establishes a relation between spread codes and
RSL codes. The proof is easy, but rather technical, hence we omit it.

\begin{proposition}\label{p:spread_lift}
  Let $n=rk$, $\F_q\subseteq \F_{q^k}\subseteq \F_{q^n}$ finite fields,
  and $P\in GL_k(\F_q)$ the companion matrix of a monic irreducible
  polynomial $p\in \F_q[x]$ of degree $k>0$. Let $\Ll\in \F_{q^k}$ be
  a root of $p$, $\mu_1,\dots,\mu_r\in \F_{q^n}$ a basis of $\F_{q^n}$
  over $\F_{q^k}$. Moreover, let $\psi:\F_{q^n}\rightarrow \F_q^n$ be
  the isomorphism of $\F_q$-vector spaces which maps the basis
  $\displaystyle (\Ll^i\mu_j)_{\substack{0\leq j\leq k-1\\ 1\leq i\leq
      r}}$ to the standard basis of $\F_{q^n}$ over $\F_q$.  Then for
  every choice of $A_0,\dots,A_{r-1}\in \F_q[P]$ there exists a unique
  linearized polynomial of the form $f=ax$ with $a\in \F_{q^n}$ such
  that
  \[(A_0 \ \cdots \  A_{r-1})=\left(
    \begin{array}{c}
      \psi(f(1))\\ \psi(f(\Ll)) \\ \vdots \\ 
      \psi(f(\Ll^{k-1}))
    \end{array}
  \right).\] The constant $a$ is $a=\psi^{-1}(v)$ where
  $v\in \F_q^n$ is the first row of $(A_0 \ \cdots \
  A_{r-1})$.
\end{proposition}

The proposition allows us to relate our spread codes to some RSL
codes. The following corollary makes the connection explicit. We use
the notation of Proposition~\ref{p:spread_lift}. 

\begin{corollary}\label{c:connection_RSL}
For each $1\leq i\leq r-1$, let $\mu_{1,i},\ldots,\mu_{r-i,i}$ be a basis of
$\F_{q^{(r-i)k}}$ over $\F_{q^k}$. Let
$\psi_i:\F_{q^{(r-i)k}}\rightarrow \F_q^{(r-i)k}$ denote the 
isomorphism of vector spaces that maps the basis 
$\displaystyle
(\Ll^j\mu_{l,i})_{\substack{0\leq j\leq k-1\\ 1\leq l\leq 
      r-i}}$ to the standard basis of $\F_q^{(r-i)k}$. 
Then
\begin{multline*}
\Svs_r=\bigcup_{i=1}^{r-1}\left\{\left.\rs\left(\underbrace{0 \ 
        \cdots \  0}_{i-1 \mbox{ times}}\  I\ 
      \begin{array}{c} \psi_i(f(1)) \\ \vdots \\
        \psi_i(f(\Ll^{k-1})) \end{array}\right)\;\right| \; f=ax, \;
  a\in \F_{q^{(r-i)k}}\right\}\\
\bigcup\left\{ \left(\underbrace{0 \ 
        \cdots \  0}_{r-1 \mbox{ times}}\  I\right)\right\}.
\end{multline*}
\end{corollary}

Corollary~\ref{c:connection_RSL} readily follows from
Proposition~\ref{p:spread_lift}. 

The connection that we have established with RSL codes allows us to
extend any minimum-distance decoding algorithm for RSL codes to a
minimum-distance decoding algorithm for spread codes. We start with a
key lemma.

\begin{lemma}\label{n=rk_RSL}
  Let $\Svs_r$ be a spread code, and $\Rvs=\rs\mat{R_1& \cdots & R_r}\in
  \Gr{\F_q}{\tilde{k}}{n}$ for some $\tilde{k}\leq k$. Assume there
  exists a $\Cvs=\rs\mat{C_1 & \cdots & C_r} \in \Svs_r$ such that
  $d(\Rvs,\Cvs)<k$. Let 
\[i:=\min\left\{j\in \{1,\dots, r\}\mid 
  \rank(R_j)>\frac{\tilde{k}-1}2\right\}.\] It holds that:
  \begin{itemize}
  \item $C_j=0$ for $1\leq j<i$,
  \item $C_i=I$, and
  \item $d(\rs\mat{R_i & R_{i+1} & \cdots & R_r},\rs\mat{I & C_{i+1} &\cdots
    & C_r})< k$.
  \end{itemize}
\end{lemma}

\begin{proof}
The result follows from Lemma~\ref{l:n=rk} and the observation
that $$d(\Cvs,\Rvs)\geq
d(\rs\mat{C_i&\cdots&C_r},\rs\mat{R_i&\cdots&R_r}).$$ 
\end{proof}

In the next proposition, we use Corollary~\ref{c:connection_RSL} and
Lemma~\ref{n=rk_RSL} to adapt to spread codes any decoding algorithm
for RSL codes. In particular, we apply our results to the algorithms
contained in~\cite{ko08} and~\cite{si08a}, and we give the complexity
of the resulting algorithms for spread codes. 

\begin{proposition}\label{p:compl}
  Any minimum-distance decoding algorithm for RSL codes may be
  extended to a minimum-distance decoding algorithm for spread
  codes. In particular, the algorithms described in~\cite{ko08}
  and~\cite{si08a} can be extended to minimum-distance decoding
  algorithms for spread codes, with complexities
  $\mathcal{O}(\F_{q^{n-k}};n^2)$ for the former and
  $\mathcal{O}(\F_{q^{n-k}};k(n-k))$ for the latter.
\end{proposition}

\begin{proof}
Suppose we are given a minimum-distance decoding algorithm for RSL
codes. We construct a minimum-distance decoding algorithm for spread
codes as follows: Let $\Rvs=\rs\mat{R_1& \cdots & R_r}\in\Gr{\F_q}{k}{n}$ be the
received word, and assume that there exists a $\Cvs=\rs\mat{C_1 &
  \cdots & C_r} \in \Svs_r$ such that $d(\Rvs,\Cvs)<k$.
First, one computes the rank of $R_1,R_2,\ldots$ until one finds an
$i$ such that $\rank(R_i)>(k-1)/2$, $\rank(R_j)\leq (k-1)/2$ for
$j<i$. Thanks to Lemma~\ref{n=rk_RSL}, one knows that $C_j=0$ for
$j<i$ and $C_i=I$. Moreover, one has 
$$
d(\rs\mat{R_i & R_{i+1} & \cdots & R_r},\rs\mat{I & C_{i+1} &\cdots
    & C_r})< k
$$ 
Therefore, one can apply the minimum-distance decoding algorithm for
RSL codes to the received word $\rs\mat{R_i & R_{i+1} & \cdots & R_r}$
in order to compute $C_{i+1},\ldots,C_r$.

Assume now that one uses as minimum-distance decoder for RSL codes
either the decoding algorithm from \cite{ko08}, or the one
from~\cite{si08a}. The complexity of computing the rank of
$R_1,\ldots, R_i$ by computing row reduced echelon forms is
$\mathcal{O}(\F_q;nk^2)$. The complexity of the decoding algorithm for RSL
codes is $\mathcal{O}(\F_{q^{n-k}};n^2)$ for the one in~\cite{ko08} and
$\mathcal{O}(\F_{q^{n-k}};k(n-k))$ for the one in~\cite{si08a}. The complexity
of the decoding algorithm is the dominant term in the complexity
estimate. 
\end{proof}
\medskip

It is well known that RSL codes are strictly related to the
rank-metric codes introduced in~\cite{ga85a}. Although the rank metric
on rank-metric codes is equivalent to the subspace distance on RSL
codes, the minimum-distance decoding problem in the former is
not equivalent to the one in the latter. In~\cite{si08a} the 
authors introduced the \emph{Generalized Decoding Problem for
  Rank-Metric Codes}, which is equivalent to the minimum-distance
decoding problem of RSL codes. Decoding algorithms for rank-metric
codes such as the ones contained in~\cite{ga85a,lo06,ri04} must be
generalized in order to be effective for the \emph{Generalized
  Decoding Problem for Rank-Metric Codes}, and consequently, to be
applicable to RSL codes. 
\medskip

Another interesting application of Lemma~\ref{n=rk_RSL} allows us to
improve the efficiency of the decoding algorithm for the codes
proposed in~\cite{sk10}. For the relevant definitions, we refer the
interested reader to the original article. 

\begin{corollary}\label{c:sk_dec}
There is an algorithm which decodes the codes from~\cite{sk10} and
has complexity $\mathcal{O}(\F_{q^{n-k}};k(n-k))$. 
\end{corollary}

The algorithm is a combination of Lemma~\ref{n=rk_RSL} and the
decoding algorithm contained in~\cite{si08a}. First, by
Lemma~\ref{n=rk_RSL}, one finds the position of the identity
matrix. This reduces the minimum-distance decoding problem to decoding
a RSL code, so one can use the algorithm from~\cite{si08a}.

\section{The Minimum-Distance Decoding Algorithm}\label{s:dec_alg}

In this section we devise a new minimum-distance decoding algorithm for
spread codes. In the next section, we show that our algorithm is more
efficient than the ones present in the literature, when $n\gg k$.

We start by proving some results on matrices, which we will be used to
design and prove the correctness of the decoding algorithm.

\subsection{Preliminary results on matrices}\label{ss:mat}

Let $\F$ be a field and let $m\in \F[y_1,\dots,y_s]$ be a polynomial
of the form $m=\sum_{U\subseteq (1,\dots,s)}a_Uy_U$ where
$y_U:=\prod_{u\in U}y_u$, $a_{(1,\dots,s)}\neq 0$.

\begin{lemma}\label{l:decomp}
The following are equivalent:
\begin{enumerate}
\item The polynomial $m$ decomposes in linear factors, i.e.,
  \[m=a_{(1,\dots,s)}\prod_{u\in (1,\dots,s)}(y_u+\mu_u)\] where
  $\mu_u=\displaystyle \frac{a_{(1,\dots,s)\setminus
      (u)}}{a_{(1,\dots,s)}}\in \F$.
\item It holds
\begin{equation}\label{e:decomp_cond}
  a_Ua_V=a_{U\cap V}a_{(1,\dots,s)}
\end{equation}
for all $U,V$ such that $|V|=s-1$ and \[\min\left((1,\dots,s)\setminus V\right) <
\min \left((1,\dots,s)\setminus U\right).\] 
\end{enumerate}
\end{lemma}

\begin{proof}

We proceed by induction on $s$.
\begin{description}
\item[\fbox {$\Rightarrow$}] If $s=1$, $m$ is a linear polynomial.  Let
  us now suppose  the thesis is true for $s-1$. Then
  \[a_{(1,\dots,s)}\prod_{u\in (1,\dots,s)}(y_u+\mu_u) =
  a_{(1,\dots,s)}(y_s+\mu_s)\left(\sum_{U\subseteq
    (1,\dots,s-1)}\tilde{a}_Uy_U\right)\] where $\tilde{a}_{(1,\dots,s-1)}=1$
  and the coefficients $\tilde{a}_U$ with $U\subseteq (1,\dots,s-1)$
  satisfy by hypothesis condition \eqref{e:decomp_cond}. The
  coefficients of $m$ are $a_U=\tilde{a}_{U\setminus (s)}$ if $s\in
  U$, and $a_U=\mu_s\tilde{a}_U$ otherwise. Therefore we only need to
  prove that \eqref{e:decomp_cond} holds for $U\in (1,\dots,s-1)$. The
  equality is $a_{(1,\dots,s)}a_U=a_Ua_{(1,\dots,s)}$ hence it is
  trivial.
 
\item[\fbox{$\Leftarrow$}] The thesis is trivial for $s=1$. Let us
  assume that the thesis holds for $s-1$. We explicitly show the
  extraction of a linear factor of the polynomial.

  \begin{align*}
    m & = \displaystyle \sum_{U\subseteq (1,\dots,s)}a_Uy_U =
    \sum_{\substack{ U\subseteq (1,\dots,s) \\ 1\in U}} \left(a_U y_U+
      a_{U\setminus (1)} y_{U\setminus (1)}\right) =\\ & =
    \displaystyle \sum_{\substack{ U\subseteq (1,\dots,s) \\
        1\in U}} \left(a_U y_1y_{U\setminus (1)}+
      \frac{a_Ua_{(2,\dots,s)}}{a_{(1,\dots,s)}} y_{U\setminus
        (1)}\right)= \\ & =
    \left(y_1+\frac{a_{(2,\dots, s)}}{a_{(1,\dots,s)}}\right)\cdot 
    \left(\sum_{\substack{ U\subseteq (1,\dots,s) \\
          1\in U}} a_Uy_{U\setminus (1)}\right).
  \end{align*} 
   The thesis is true by induction. 
\end{description}
\end{proof}

\medskip

Let $\F[x_{i,j}]_{1\leq i,j\leq k}$ be a polynomial ring with
coefficients in a field $\F$. Consider the generic matrix of size
$k\times k$ 
\[M:=\mat{x_{1,1} & \cdots & x_{1,k} \\ \vdots & & \vdots \\
  x_{k,1} & \cdots & x_{k,k}}.\]
Denote by $\II_{s+1}\subset \F[x_{i,j}]_{1\leq i,j\leq n}$ the ideal
generated by all minors of size $s+1$ of $M$, which do not involve
entries on the diagonal, i.e.,
\[\II_{s+1}:=([J,L]\mid J,L\in\{1,\dots,k\}^{s+1}, \ J\cap L=
\emptyset).\]
We establish some relations on the minors of $M$, modulo the
ideal $\II_{s+1}$.

\begin{lemma}\label{l:minors}
  Let $J=(j_1,\dots,j_k),L=(l_1,\dots,l_k)\in \{1,\dots,k\}^k$,
  $J_s=(j_1,\dots,j_s)$, and $L_s=(l_1,\dots,l_s)$. Then
  \[[J_s;L_s][J;L]=\sum_{t=s+1}^k (-1)^{t+s+1}[J_s\cup (j_t);L_s\cup
  (l_{s+1})][J\setminus (j_t) ;L\setminus (l_{s+1})].\]
\end{lemma}

\begin{proof}
  Notice that if we consider as convention that
  $[\emptyset;\emptyset]=1$, i.e., when $s=0$, we get the determinant
  formula.

  We proceed by induction on $s$. Let us consider the case when $s=1$,
  i.e., $[J_1;L_1]=\mat{x_{j_1,l_1}}$. Then,
\begin{eqnarray*}
  \mat{x_{j_1,l_1}}[J;L] & = &
  \sum_{t=1}^k(-1)^{t+2}x_{j_1,l_1}x_{j_t,l_2}[J\setminus
  (j_t);L\setminus (l_2)]\\
  & = & -x_{j_1,l_1}x_{j_1,l_2}[J\setminus  (j_1);L\setminus (l_2)] \\
  & & +
  \sum_{t=2}^k (-1)^{t+2}\left([(j_1,j_t);(l_1,l_2)]+x_{j_t,l_1}x_{j_1,l_2}\right)
  [J\setminus (j_t);L\setminus (l_2)]\\  
  & = & \sum_{t=2}^k(-1)^{t+2}[(j_1,j_t);(l_1,l_2)] [J\setminus
  (j_t);L\setminus (l_2)] \\
  & & + 
  x_{j_1,l_2}[J;(l_1,l_1,l_3,\dots,l_k)]. 
\end{eqnarray*}
For $s=1$ the thesis is true because $[J;(l_1,l_1,l_3,\dots,l_k)]=0$
since column $l_1$ appears twice.

Assume that the thesis is true for $s-1$. 
\begin{align*}
  [J_s;L_s][J;L]  = 
  \sum_{t=1}^k(-1)^{t+s+1}x_{j_t,l_{s+1}}[J_s;L_s][J\setminus
   (j_t);L\setminus  (l_{s+1})].
\end{align*}
Let us now focus on the factor $x_{j_r,l_{s+1}}[J_s;L_s]$ for $r\geq s+1$,
we get
\begin{align*}
  x_{j_r,l_{s+1}}[J_s;L_s]& = [J_s\cup  (j_r);L_s\cup  (l_{s+1})]\!+
  \!\!\sum_{t=1}^{s}(-1)^{t+s}x_{j_t,l_{s+1}}[J_s\setminus  (j_t)\cup  (j_r) ;L_s]. 
\end{align*}
By substitution it follows that 
\begin{eqnarray*}
  [J_s;L_s][J;L] & = &\sum_{t=s+1}^k (-1)^{t+s+1}[J_s\cup
  (j_t);L_s\cup (l_{s+1})]
  [J\setminus (j_t);L\setminus  (l_{s+1})]+\\
  & & + \sum_{t=1}^s (-1)^{t+s+1}x_{j_t,l_{s+1}}\left(
    [J_s;L_s][J\setminus
    (j_t);L\setminus  (l_{s+1})]+\phantom{\sum_{j=s+1}^k(-1)^{j+s}}\right.\\
  & &\left. +\sum_{r=s+1}^k(-1)^{r+s} [J_s\setminus (j_t)\cup
    (j_r);L_s][L\setminus
    (j_r);L\setminus  (l_{s+1})]\right) \\
  & = & \sum_{t=s+1}^k (-1)^{t+s+1}[J_s\cup (j_t);L_s\cup (l_{s+1})]
  [J\setminus (j_t);L\setminus  (l_{s+1})]+\\
  & &+ \sum_{t=1}^s (-1)^{t+s+1}x_{j_t,l_{s+1}} \left( 
    [J_s\setminus (j_t);L_s\setminus (l_s)][J;\bar{L}]\right)
\end{eqnarray*}
where $\bar{L}=(l_1,\dots,l_s,l_s,l_{s+2},\dots,l_k)$. The
repetition of column $l_s$ twice in $\bar{L}$ implies that
$[J;\bar{L}]=0$. The last equality follows from the induction
hypothesis.
\end{proof}

The following is an easy consequence of Lemma~\ref{l:minors}.

\begin{proposition}\label{p:rel_ideal}
  Let $J,L\subset K=(1,\dots,k)$ such that $J\cap L=\emptyset$. Then
  \[[J,L][K,K]-[J\cup (i);L\cup (i)][K\setminus (i);K\setminus
  (i)]=\sum_{l\in K\setminus (J\cup (i))}h_l\Min{J\cup (i)}{L\cup
    (l)}{}\in \II_{s+1},\] with $h_l\in \F[x_{i,j}]_{1\leq i,j\leq
    k}$ for any $l\in K\setminus (J\cup (i))$.
\end{proposition}

\medskip

We now study the minors of a matrix of the form $S^{-1}NS$ where
$N\in\F_q^{k\times k}$ and $S$ has a special form, which we describe
in the next lemma.

\begin{lemma}\label{l:S} 
  Let $P\in GL_k(\F_q)$ to be the companion matrix of a monic
  irreducible polynomial $p\in \F_q$ of degree $k>0$, and let $\Ll\in
  \F_{q^k}$ be a root of $p$. Then the matrix
  \begin{equation}\label{e:transf_mat}
    S:=\mat{1 & 1 & 1 & \cdots & 1 \\
      \Ll & \Ll^{[1]} & \Ll^{[2]} & \cdots & \Ll^{[k-1]} \\
      \Ll^2 & \Ll^{2\cdot [1]} & \Ll^{2\cdot [2]} & \cdots &
      \Ll^{2\cdot [k-1]} \\
      \vdots & \vdots & \vdots & & \vdots \\
      \Ll^{k-1} & \Ll^{(k-1)\cdot [1]} & \Ll^{(k-1)\cdot [2]} & \cdots &
      \Ll^{(k-1)\cdot [k-1]}}. 
  \end{equation}
  diagonalizes $P$.
\end{lemma}

\begin{proof}
  The eigenvalues of the matrix $P$ correspond to the roots of the
  irreducible polynomial $p\in \F_q[x]$. If $\Ll\in \F_{q^k}$ is an
  element such that $p(\Ll)=0$, then
  $p=\prod_{i=0}^{k-1}(x-\Ll^{[i]})$ by \cite[Theorem~2.4]{li94}. It
  is enough to show that the columns of $S$ correspond to the
  eigenvectors of $P$. Let $i\in \{0,\dots,k-1\}$, then
  \begin{align*}
    P\mat{1 \\ \Ll^{[i]} \\ \vdots \\ \Ll^{(k-1)\cdot [i]}}& =
    \mat{\Ll^{[i]} \\ \Ll^{2\cdot [i]} \\ \vdots \\
      -\sum_{j=0}^{k-1}p_j\Ll^{j\cdot [i]}}=\mat{\Ll^{[i]}
      \\ \Ll^{2\cdot [i]} \\ \vdots \\
      \left(-\sum_{j=0}^{k-1}p_j\Ll^{j}\right)^{[i]}}\\
    & = \mat{\Ll^{[i]}
      \\ \Ll^{2\cdot [i]} \\ \vdots \\
      \Ll^{k \cdot [i]}} = \Ll^{[i]}\mat{1 \\ \Ll^{[i]} \\ \vdots \\
      \Ll^{(k-1)\cdot [i]}}.
  \end{align*}
\end{proof}

We now establish some properties of $S$.

\begin{lemma}\label{l:SS-1}
  The matrices $S$ and $S^{-1}$ defined by \eqref{e:transf_mat}
  satisfy the following properties:
  \begin{enumerate}
  \item the entries of the first column of $S$ (respectively, the
    first row of $S^{-1}$) form a basis of $\F_{q^k}$ over $\F_q$, and
  \item the entries of the $(i+1)$-th column of $S$ (respectively,
    row of $S^{-1}$) are the $q$-th power of the ones of the $i$-th
    column (respectively, row) for $i=1,\dots,k-1$.
  \end{enumerate}
\end{lemma}

\begin{proof}
  The two properties for the matrix $S$ come directly from its
  definition. By \cite[Definition~2.30]{li94} we know that there
  exists a unique basis $\{\Gl_0,\dots, \Gl_{k-1}\}$ of $\F_{q^k}$
  over $\F_q$ such that
  \[\mathrm{Tr}_{\F_{q^k}/\F_q}(\Ll^i\Gl_j)=\left\{ \begin{array}{r}
      1 \quad i= j\\ 0 \quad i\neq j \end{array}\right.,\]
  where $\mathrm{Tr}_{\F_{q^k}/\F_q}(\Al):= 1+\Al^{[1]}+\dots +\Al^{[k-1]}$
  for $\Al\in \F_{q^k}$. 
  We have
  \[ S^{-1}=\mat{\Gl_0 & \Gl_1 & \cdots & \Gl_{k-1} \\
    \Gl_0^{[1]} & \Gl_1^{[1]} & \cdots & \Gl_{k-1}^{[1]} \\
    \vdots & \vdots & & \vdots \\ 
    \Gl_0^{[k-1]} & \Gl_1^{[k-1]} & \cdots & \Gl_{k-1}^{[k-1]}}.\]
\end{proof}

The next theorem and corollary will be used in
Subsection~\ref{ss:non_sing} to devise a simplified minimum-distance
decoding algorithm, under the assumption that the first $k$ columns of
the received vector space are linearly independent.

\begin{theorem}\label{t:reg_mat}
  Let $t\leq k$ and let $N\in \F_q^{t\times k}$ and $S\in \F_{q^k}^{k\times
    t}$ be two matrices satisfying the following properties:
\begin{itemize}
\item $N$ has full rank,
\item the entries of the first column of $S$ form a basis of
  $\F_{q^k}$ over $\F_q$, and
\item the entries of the $(i+1)$-th column of $S$ are the $q$-th
  power of the ones of the $i$-th column, for $i=1,\dots,t-1$. 
\end{itemize}
Then $NS\in GL_t(\F_{q^k})$.
\end{theorem}

\begin{proof}
  Let 
  \[N:=\left(n_{ij}\right)_{\substack{1\leq i \leq t
      \\ 1\leq j \leq k}}\mbox{ and
  }NS=\left(t_{ij}\right)_{\substack{1\leq i \leq t \\
      1\leq j \leq t}}.\] 
  Let $\displaystyle S:=\left(s_{ij}\right)_{\substack{1\leq
      i \leq k \\ 1\leq j \leq
      t}}=\left(s_i^{[j-1]}\right)_{\substack{1\leq i \leq k \\
      1\leq j \leq t}}$ where $s_1,\dots,s_k\in \F_{q^k}$ form a basis
  of $\F_{q^k}$ over $\F_q$. Then:
  \[t_{ij}:=\sum_{l=1}^{k}n_{il}s_{lj}=\sum_{l=1}^{k}n_{il}s_l^{[j-1]}=
  \left(\sum_{l=1}^{k}n_{il}s_l\right)^{[j-1]},\] since the entries
  of $N$ are in $\F_q$. Let $\tau_i:=\sum_{l=1}^kn_{il}s_l\in
  \F_{q^k}$, then 
  \[NS=\mat{\tau_1 & \tau_1^{[1]} &\dots & \tau_2^{[t-1]} \\
    \tau_2 & \tau_1^{[1]} &\dots & \tau_2^{[t-1]} \\
    \vdots & \vdots & & \vdots \\
    \tau_r & \tau_r^{[1]} &\dots & \tau_r^{[t-1]} }.\] The elements
  $\tau_1,\dots,\tau_t\in \F_{q^k}$ are linearly independent over
  $\F_q$. Indeed, the linear combination
  \begin{eqnarray*}\sum_{i=1}^k\Al_i \tau_i & = &
    \sum_{i=1}^t\Al_i \sum_{l=1}^k n_{il}s_l =  \sum_{l=1}^k
    \left(\sum_{i=1}^k\Al_i n_{il}\right)s_l
  \end{eqnarray*}
  is zero only when $\sum_{i=1}^t\Al_i n_{il}=0$ for
  $l=1,\dots,t$. Since $N$ has full rank it follows that $\Al_1,\dots
  ,\Al_t$ must all be zero, leading to the linear independence of
  $\tau_1,\dots,\tau_t$.

  Now let $a_0,\dots, a_{t-1} \in \F_{q^k}$ be such that
  \[NS\mat{a_0 \\ \vdots \\ a_{t-1}} = 0,\]
  and consider the linearized polynomial
  $f=\sum_{i=0}^{t-1}a_ix^{[t-i]}$. The elements
  $\tau_1,\dots,\tau_t$ are by assumption roots of $f$. Since $f$
  is a linear map, the kernel of $f$ contains the subspace $\langle
  \tau_1,\dots,\tau_t\rangle\subset \F_{q^k}$. Therefore $f$ is a
  polynomial of degree $q^{t-1}$ with $q^t$ different roots, then
  $a_0=\cdots=a_{t-1}=0.$
\end{proof}

\begin{corollary}\label{c:consecutive}
  Let $S\in GL_k(\F_{q^k})$ be the matrix specified in
  \eqref{e:transf_mat} and $N\in \F_q^{k\times k}$. Then for any
  $J,L\subset(1,\dots,k)$ tuples of consecutive indices with
  $|J|=|L|=\rank(N)$, one has $[J;L]_{S^{-1}NS}\neq 0.$
\end{corollary}

\begin{proof}
  Let $t=\rank (N)$ and $J,L\subset (1,\dots,k)$ with $|J|=|L|=t$, let
  $H=(1,\dots,t)$. Let $N_1\in \F_q^{k\times t}$ and $N_2\in 
  \F_q^{t\times k}$ be matrices with full rank such that
  $N=N_1N_2$. One has
\[\Min{J}{L}{S^{-1}NS}=\Min{J}{L}{S^{-1}N_1\cdot N_2S}=
\Min{J}{H}{S^{-1}N_1}\Min{H}{L}{N_2S}.\] 

We can now focus on the characterization of the maximal minors of the
matrix $N_2S$. The following considerations will also work 
for the matrix $S^{-1}N_1$ considering its transpose.

The minor $\Min{H}{L}{N_2S}$ is the determinant of a square matrix
obtained by multiplying $N_2$ with the submatrix consisting of the
columns of $S$ indexed by $L$. Let $L$ contain consecutive indices. By
Lemma \ref{l:SS-1}, the submatrix of $S$ that we obtain together with
$N_2$ satisfy the conditions of Theorem \ref{t:reg_mat}. It follows
that $\Min{H}{L}{N_2S}\neq 0$.

As a consequence we have that $\Min{J}{L}{S^{-1}NS}\neq 0$ when
both $J$ and $L$ are tuples of consecutive indices.
\end{proof}

The following is a reformulation of Corollary \ref{c:consecutive} for
small rank matrices.

\begin{corollary}\label{t:ndrank}
Let $N\in \F_{q}^{k\times k}$ be a matrix such that $\rank(N)\leq
\frac{k-1}2$ and $S\in GL_k(\F_{q^k})$ defined as in
\eqref{e:transf_mat}. 
Then for any choice $J,L\subset (1,\dots,k)$ of consecutive indices
with $|J|=|L|=\rank(N)$,
\[\Min{J}{L}{S^{-1}NS}\neq 0.\]
In particular, $$\mathrm{ndrank}(S^{-1}NS)=\rank(N).$$
\end{corollary}
  
\subsection{The Decoding Algorithm}

In this subsection we devise an efficient minimum-distance decoding
algorithm for spread codes, and establish some closely related
mathematical results. 

We start by reducing the minimum-distance decoding algorithm for
$\Svs_r$ to at most $r-1$ instances of the minimum-distance
decoding algorithm for $\Svs_2$. Notice that the minimum-distance
decoders for the case $r=2$ can be run in parallel.

Let $\Rvs=\rs\mat{R_1 & \cdots & R_r}$ be a received
  space. We assume that $$1\leq \tilde{k}=\rank(\Rvs)\leq k.$$

Algorithm~\ref{a:dec_alg_n=rk} on page~\pageref{a:dec_alg_n=rk} is
based on the following lemma. 

\begin{lemma}\label{l:n=rk}
  Let $\Svs_r$ be a spread code, and $\Rvs=\rs\mat{R_1& \cdots &
    R_r}\in \Gr{\F_q}{\tilde{k}}{rk}$. Assume there exists a
  $\Cvs=\rs\mat{C_1 & \cdots & C_r}\in \Svs$ such that $d(\Rvs,\Cvs)
  <k$. It holds
  \[C_i=0 \iff \rank(R_i)\leq \frac{\tilde{k}-1}2.\]
\end{lemma}

\begin{proof}
  \begin{description}
\item[\fbox {$\Rightarrow$}] Let $i\in \{1,\dots,r\}$ be an index
  such that $C_i=0$. By the construction of a spread code there
  exists a $j\in \{1,\dots,r\}$ with $C_j=I$. We claim that $\dim
  (\Cvs\cap\Rvs)>\frac{\tilde{k}}2$. In fact, $$k>\dim\Cvs+\dim\Rvs-2\dim
  (\Cvs\cap\Rvs)=k+\tilde{k}-2\dim(\Cvs\cap\Rvs).$$ From the claim it
  follows that
  \[\rank \mat{0 &I \\ R_i& R_j}\leq \rank \mat{
      C_1 & \cdots & C_r \\ R_1 & \cdots &
      R_r}=k+\tilde{k}-\dim(\Cvs\cap\Rvs) < k+\frac{\tilde{k}}2.\]
 This proves that \[\rank(R_i) < \frac{\tilde{k}}2.\] 

\item[\fbox {$\Leftarrow$}] Let $i\in \{1,\dots,r\}$ be such that
  $\rank(R_i)\leq \frac{\tilde{k}-1}2$ and assume by contradiction that
  $C_i\in \F_q[P]^*$. It follows that
  \[\dim (\Cvs\cap\Rvs)\leq\dim (\rs(C_i)\cap\rs(R_i))=\dim (\rs(R_i))\leq
  \frac{\tilde{k}-1}2 \] 
  which contradicts the assumption that $d(\Cvs,\Rvs)=k+\tilde{k}-2\dim(\Cvs
  \cap \Rvs)<k$.
\end{description}
\end{proof}

\begin{algorithm}\label{a:dec_alg_n=rk}
  \caption{Minimum-distance decoding algorithm: $n=rk$, $r>2$}
  \Input{$\Rvs=\rs\mat{R_1 & \cdots & R_r}\in \Gr{\F_q}{\tilde{k}}{rk}$,
    $r>2$, \\
    $P\in GL_k(\F_q)$ the companion matrix of $p\in \F_q[x]$ and \\
    $S\in GL_k(\F_{q^k})$ its diagonalizing matrix.}  

  \Output{$\Cvs \in\Svs_r \subset \Gr{\F_q}{k}{rk}$ such that
    $d(\Rvs,\Cvs)<k$, if such a $\Cvs$ exists.}
  

  Let $r_i=\rank(R_i)$ for $i=1,\dots,r$\; 
  

  \If{$r_i\leq \frac{\tilde{k}-1}2$ for all $i\in\{1,\dots,r\}$}{\Return there exists
    no $\Cvs\in \Svs_r$ such that $d(\Rvs,\Cvs)<k$ }

  
  Let $j=\min\left\{i\in \{1,\dots r\} \mid r_i>\frac{\tilde{k}-1}2\right\}$\;

  
  \For{$i\in \{1,\dots,r\}$ and $r_i\leq\frac{\tilde{k}-1}2$}{$C_i=0\in
    \F_q^{k\times k}$\;}

  
  \For{$j<i\leq r$ and $r_i>\frac{k-1}2$}{
    Run a minimum-distance decoding algorithm for $r=2$ with input
    $\Rvs=\rs\mat{R_j & R_i}$, $P$ and $S$\;
    \If{minimum-distance decoding algorithm returns no
      $\Cvs\in\Svs_2$}{\Return there exists no $\Cvs\in \Svs_r$ such
      that $d(\Rvs,\Cvs)<k$\; 
    \lElse{let $C_i\in \F_q[P]$ such that $\Cvs=\rs\mat{I & C_i}$\;}}
  }


  \Return $\Cvs=\rs\mat{C_1 & \cdots &C_r}$.
\end{algorithm}

\medskip

Because of Lemma~\ref{l:n=rk}, we may now focus on designing a
minimum-distance decoding algorithm for the case where $n=2k$. For the
remainder of this subsection, we consider the spread code
\begin{equation*}
\Svs=\Svs_2=\left\{ \rs\mat{I & A}\mid A\in
  \F_{q}[P]\right\} \cup \left\{\rs\mat{0 & I} \right\}
\end{equation*}
where $I$ and $0$ are respectively the identity and the zero
matrix of size $k \times k$.

Since a minimum-distance decoding algorithm decodes uniquely up to half the
minimum distance, we are interested in writing an algorithm with the
following specifications:
\begin{description}
\item[Input]  $\Rvs=\rs\mat{R_1 & R_2}\in \Gr{\F_q}{\tilde{k}}{2k}$, \\
    $P\in GL_k(\F_q)$ the companion matrix of $p\in \F_q[x]$ and \\
    $S\in GL_k(\F_{q^k})$ its diagonalizing matrix.
\item[Output] $\Cvs \in\Svs \subset \Gr{\F_q}{k}{2k}$ such that
  $d(\Rvs,\Cvs)<\frac{d(\Svs)}2=k$, if such a $\Cvs$ exists.
\end{description}

We first give a membership criterion for spread codes. We follow the
notation given at the beginning of this section.

\begin{proposition}[{\cite[Lemma~5 and Corollary~6]{ma08p}}]\label{diag}
  Let $A\in GL_k(\F_q)\cup \{0\}$. The following are equivalent:
  \begin{enumerate}
  \item $A\in \F_q[P]$.
  \item $S^{-1}AS$ is a diagonal matrix.
  \item $AP=PA$.
  \end{enumerate}
  If this is the case, then $S^{-1}AS=\Delta(\Ll)$ for some $\Ll\in
  \F_{q^k}$.
\end{proposition}

From the proposition, we get an efficient algorithm to test whether a
received vector space is error-free. 

\begin{corollary}[{\cite[Corollary~6]{ma08p}}, Membership Test]\label{membership}
  Let $\Rvs=\rs\mat{R_1 & R_2}\in \Gr{\F_q}{k}{2k}$. Then $\Rvs\in \Svs$ if and only
  if either $R_1\in GL_k(\F_q)$ and $S^{-1}R_1^{-1}R_2S$ is diagonal
  or  $R_1=0$ and $R_2\in GL_k(\F_q)$.
\end{corollary}

The following is an easy consequence of Lemma~\ref{l:n=rk}. It allows
us to efficiently test whether the sent codeword was $\rs\mat{0 & I}$, or
$\rs\mat{I & 0}$.

\begin{corollary}\label{p:infinity}
Let $\Rvs=\rs\mat{R_1 & R_2}\in \Gr{\F_q}{\tilde{k}}{2k}$ be a
received space, and assume that it is uniquely decodable. The
following are equivalent: 
\begin{itemize}
\item $\rank (R_1)\leq \frac{\tilde{k}-1}2$, and
\item the output of a minimum-distance decoding algorithm is
  $\rs\mat{0 & I}$. 
\end{itemize}
The analogous statement holds for $R_2$.
\end{corollary}

Because of Corollary~\ref{p:infinity}, we can restrict our decoding
algorithm to look for codewords of the form $\Cvs = \rs \mat{I & A}$
where $A\in \F_q[P]$.  Since there is an obvious symmetry in the
construction of a spread code, we assume without loss of generality
that 
\[\rank (R_1) \geq \rank (R_2) > \frac{\tilde{k}-1}2.\]

\medskip

With the following theorem we translate the unique decodability
condition into a rank condition, and then into a
greatest common divisor condition.

\begin{theorem}\label{t:unique_dec_sol}
  Let $\Rvs\in \Gr{\F_q}{\tilde{k}}{n}$ be a subspace with \[\rank(R_1)\geq
  \rank(R_2) > \frac{\tilde{k}-1}2.\] The following are equivalent:
  \begin{itemize}  
  \item $\Rvs$ is uniquely decodable.
  \item There exists a unique $\mu\in \F_{q^k}$ such that
    \begin{equation}\label{e:rank_cond}
      \rank(S^{-1}R_1S\Delta(\mu)-S^{-1}R_2S)\leq\frac{\tilde{k}-1}2
    \end{equation}
  \item 
    $x-\mu=\gcd\left(\left\{[J;L]_{S^{-1}R_1S\Delta(x)-S^{-1}R_2S}
        \mid J,L\in \{1,\dots,k\}^{\lfloor\frac{\tilde{k}+1}2\rfloor}
      \right\},x^{[k]}-x\right)$, for some $\mu\in \F_{q^k}$.
  \end{itemize}
\end{theorem}

\begin{proof}
  $\Rvs$ is uniquely decodable if and only if there exists a unique
  matrix $X\in \F_q[P]$ such that
  \begin{eqnarray*} k-1\geq d(\Rvs,\Cvs) & = & 2 \rank \mat{I&X\\R_1&R_2}
    -(k+\tilde{k}) \\ = 2\rank \mat{I & X\\0 & R_1X-R_2}-(k+\tilde{k})
  & = & 2\rank(R_1X-R_2)+k-\tilde{k}.
  \end{eqnarray*}
  Furthermore we get that $\rank (R_1X-R_2)=\rank
  (S^{-1}R_1S\Delta(x)-S^{-1}R_2S)$, where $S^{-1}XS=\Delta(x)$ is a
  consequence of Lemma \ref{diag}. The existence of a unique solution
  $X\in \F_q[P]$ is then equivalent to the existence of a unique
  $\mu\in \F_{q^k}$ such that
  \[\rank (S^{-1}R_1S\Delta(\mu)-S^{-1}R_2S)\leq \frac{\tilde{k}-1}2.\]

  This is equivalent to the condition that all minors of size
  $\lfloor\frac{\tilde{k}+1}2\rfloor$ of $S^{-1}R_1S\Delta(\mu)-S^{-1}R_2S$
  are zero. This leads to a nonempty system of polynomials in the
  variable $x$ having a unique solution $\mu\in \F_{q^k}$. Therefore
  \[x-\mu \mid \gcd
  \left(\left\{[J;L]_{S^{-1}R_1S\Delta(x)-S^{-1}R_2S} \mid J,L\in
      \{1,\dots,k\}^{\lfloor\frac{\tilde{k}+1}2\rfloor}\right\},
    x^{[k]}-x\right).\] Equality follows from the uniqueness of $\mu$.
\end{proof}

Theorem~\ref{t:unique_dec_sol} has the following immediate
consequence, which constitutes a step forward towards the design of
our decoding algorithm. 

\begin{corollary}\label{mu}
  Assume that the received space $\Rvs\in \Gr{\F_q}{\tilde{k}}{n}$ is uniquely
  decodable. Then it decodes to \[\Cvs=\rs\mat{I &
    S\Delta(\mu)S^{-1}}\in \Svs\] where $\mu\in\F_{q^k}$ is a root of
  all the minors of size $\lfloor\frac{\tilde{k}+1}2\rfloor$ of
  $S^{-1}R_1S\Delta(x)-S^{-1}R_2S$. 
\end{corollary}

Under the unique decodability assumption, decoding a received space
$\Rvs$ corresponds to computing the $\mu$ from
Corollary~\ref{mu}. However, computing the greatest common divisor of
all the minors of $S\Delta(x)S^{-1}$ of the appropriate size does not
constitute an efficient algorithm. 

The following theorem provides a significant computational
simplification of this approach. In the proof, we give a procedure to
construct one minor of $S\Delta(x)S^{-1}$ of the appropriate size,
whose factorization we can explicitly describe. In particular, we
give explicit formulas for its roots. In practice, one
wants to proceed as follows: First, find such a minor and write down
all of its roots, and second, for each root $\mu$ check whether
$\rank(S\Delta(\mu)S^{-1})\leq \lfloor\frac{\tilde{k}-1}2\rfloor.$

\begin{theorem}\label{t:suitable_poly}
  Let $\Rvs=\rs\mat{R_1& R_2}\in \Gr{\F_q}{\tilde{k}}{2k}$ be uniquely
  decodable with $\rank(R_1)\geq \rank(R_2) >
  \frac{\tilde{k}-1}2$, $S\in GL_k(\F_{q^k})$ a matrix diagonalizing $P$ and
  $M\in GL_k(\F_{q^k})$ such that $MS^{-1}\mat{R_1& R_2}S$ is in row
  reduced echelon form. Let
  $R(x):=MS^{-1}R_1S\Delta(x)-MS^{-1}R_2S$. Then, there exist
  $J,L\subset I:=(1,\dots,k)$ with $|J|=|L|=\lfloor
  \frac{\tilde{k}+1}2\rfloor-(\tilde{k}-\rank(R_1))$ such that
  \[[J;L]_{R(x)}=\mu \prod_{i\in K} (x^{[i]}-\mu_i),\] where
  $K=J\cap L$, $\mu=[J\setminus K;L\setminus K]_{R(0)}\in \F_{q^k}^*$
  and $\displaystyle \mu_i=\frac{[J\setminus (i);L\setminus
    (i)]_{R(0)}}{[J\setminus K;L\setminus K]_{R(0)}}\in \F_{q^k}$. In
  particular if $\mu\in \F_{q^k}$ is such that $\rank (R(\mu))\leq
  \frac{\tilde{k}-1}2$, then
  \[\mu\in \left\{\mu_i^{[k-i]}\mid i\in K\right\}.\]
\end{theorem}

\begin{proof}
  We first focus on the form of the matrix $R(x)$. Let
  $r_i:=\rank(R_i)$ for $i=1,2$. We deduce by Corollary
  \ref{c:consecutive} that the pivots of the matrix $MS^{-1}\mat{R_1 &
    R_2}S$ are contained in the first $r_1$ columns and, since $\dim
  \Rvs =\tilde{k}$, in a choice of $\tilde{k}-r_1$ of the first $r_2$
  columns of $MS^{-1}R_2S$. Figure \ref{f:f1} and Figure \ref{f:f2} at
  page \pageref{f:f1} depict respectively the matrix $MS^{-1}\mat{R_1
    & R_2}S$ and $R(x)$.

\medskip

\begin{figure}[htbp]
  \begin{center}
    \input{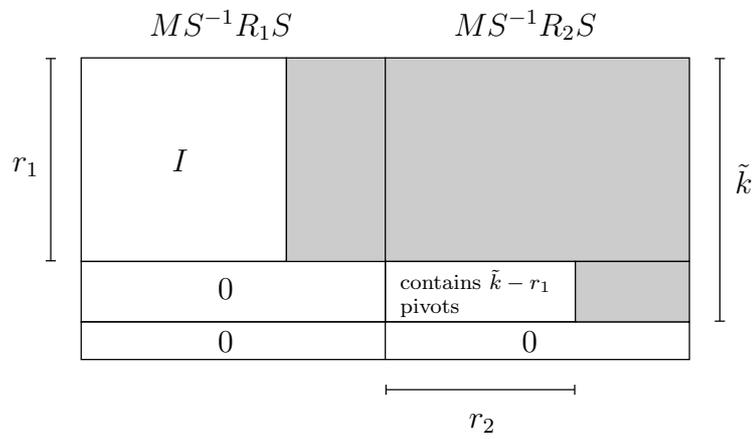}
  \end{center}
\caption{Representation of the matrix $MS^{-1}(R_1 \quad R_2)S$.}\label{f:f1}
\end{figure}

  
\begin{figure}[htbp]
  \begin{center}
    \input{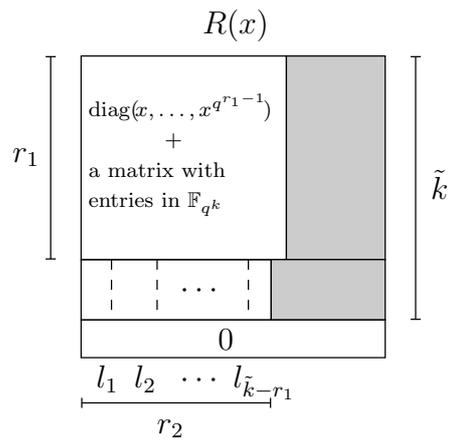}
\caption{Representation of the matrix $R(x)$.}\label{f:f2}
  \end{center}
\end{figure}

  $(l_1,\dots,l_{\tilde{k}-r_1})\subset I$ is the tuple of indices of
  the columns cor\-res\-pon\-ding to the pivots of $MS^{-1}R_2S$. Hence, for
  all $i\in \{1,\dots,\tilde{k}-r_1\}$ the entries of columns $l_i$ of $R(x)$
  are all zero except for the entry $l_i$, which is $x^{[l_i-1]}$,
  and the entry $r_1+i$, which is $1$.

  Now consider the square submatrix $R'(x)$ of $R(x)$ of size $2r_1-\tilde{k}$
  defined by the rows and columns indexed by
  \[I':=I\setminus(l_1,\dots,l_{\tilde{k}-r_1},r_1+1,\dots,k).\] 
  The matrix $R'(x)$ is a matrix containing unknowns only in the diagonal
  entries. 

  Let $\subm{J}{L}{R'(x)}$ be a submatrix of $R'(x)$, then it holds that
  \[\Min{J}{L}{R'(x)}=\Min{J\cup (r_1+1,\dots,\tilde{k})}{L\cup (l_1,\dots,
    l_{\tilde{k}-r_1})}{R(x)}.\]


  Let $\mu\in \F_{q^k}$ be the unique element satisfying condition
  \eqref{e:rank_cond}, by the previous relation it holds that
  \begin{equation}\label{e:2r_1}\rank(R'(\mu))\leq
    \frac{\tilde{k}-1}2-(\tilde{k}-r_1)=\frac{2r_1-\tilde{k}-1}2.\end{equation} 
  This implies that $\mu$ is a root of all $\Min{J}{L}{R'(x)}$ such that
  $|J|=|L|=\lfloor\frac{\tilde{k}+1}2\rfloor-(\tilde{k}-r_1)$.

  Let $J',L'\subset I'$ be tuples of indices such that
  \begin{eqnarray}\label{e:cond_alg}
    &J'\cap L'=\emptyset, \mbox{ }
    \Min{J'}{L'}{R'(x)}\neq 0 \mbox{, and } \nonumber\\ &\Min{J'\cup
      (j)}{L'\cup (l)}{R'(x)}=0\mbox{ for any } j\neq l\in I'\setminus
    (J'\cup L').
  \end{eqnarray} 
  The existence of a couple of tuples satisfying these conditions is
  ensured by the definition of $\mathrm{ndrank}(R'(x))$.  

  Let $K\subset I'\setminus(J'\cup L')$ with $|K|=\lfloor \frac{\tilde{k}+1}2
  \rfloor-(\tilde{k}-r_1)-|J'|$. $K$ is non empty since by \eqref{e:2r_1} 
  \[|K|\geq \lfloor \frac{\tilde{k}+1}2 \rfloor - (\tilde{k}-r_1) -
  \frac{2r_1-\tilde{k}-1}2=\lfloor
  \frac{\tilde{k}+1}2\rfloor-\frac{\tilde{k}-1}2 >0.\] Define
  $J:=J'\cup K \mbox{ and } L:=L'\cup K.$

  Combining conditions \eqref{e:cond_alg} and Proposition~\ref{p:rel_ideal} we
  obtain that
  \[\Min{J}{L}{}\Min{I'}{I'}{}-\Min{J\cup (i)}{L\cup
    (i)}{}\Min{I'\setminus (i)}{I'\setminus (i)}{}=0\]
  for $i\in K$. It follows by Lemma \ref{l:decomp} that the
  polynomial $\Min{J}{L}{}$ factors as follows
  \[\Min{J}{L}{R(x)}=\Min{J\setminus K}{L\setminus K}{R(0)}\prod_{i\in K}
  \left( x^{[i]}-\mu_i\right).\]
    with $\mu_i=\frac{\Min{J\setminus (i)}{L\setminus (i)}{R(0)}}{
      \Min{J\setminus K}{L\setminus K}{R(0)}}$ and $\mu\in
  \left\{\mu_i^{[k-i]}\mid i\in K\right\}$. 
\end{proof}

Summarizing, the decoding algorithm that we obtain exploiting
Theorem~\ref{t:suitable_poly} is as follows:
\begin{enumerate}
\item Find tuples $J,L$ satisfying the assumptions \eqref{e:cond_alg}
  of the theorem. Algorithm \ref{a:alg_min} in Section
  \ref{s:algorithm} gives an efficient way to find such tuples. 
\item Write down the roots of the minor $[J,L]_{R(x)}$, where $R(x)$ is
  the matrix in the statement of the
  theorem. Theorem~\ref{t:suitable_poly} gives explicit formulas for
  the roots, so this step requires a negligible amount of computation.
\item For each root $\mu$ found in the previous step, check whether
  the rank of $R(\mu)$ is smaller than or equal to
  $\lfloor\frac{\tilde{k}-1}2\rfloor$. 
\item If the unique decodability assumption is satisfied, exactly one
  root $\mu$ will satisfy the rank condition in the previous
  step. In this case, we decode to $\rs\mat{I & S\Delta(\mu)S^{-1}}$.
\item Else, none of the roots will. In this case, we have a decoding
  failure. 
\end{enumerate}

\medskip

We now give the detailed minimum-distance decoding algorithm in pseudocode. 

\begin{algorithm}\label{a:dec_alg}
  \caption{Minimum-distance decoding algorithm: $n=2k$}
  \Input{$\Rvs=\rs\mat{R_1 & R_2}\in \Gr{\F_q}{\tilde{k}}{2k}$ with
    $\mat{R_1 & R_2}\in \F_q^{k\times n}$,\\
    $P\in GL_k(\F_q)$ the companion matrix of $p\in \F_q[x]$ and \\
    $S\in GL_k(\F_{q^k})$ its diagonalizing matrix.} 
  \Output{$\Cvs \in\Svs \subset \Gr{\F_q}{k}{n}$ such that
    $d(\Rvs,\Cvs)<k$, if such a $\Cvs$ exists.}
  \BlankLine
  Let $r_i:=\rank(R_i)$ for $i=1,2$.\\
  \Begin(1.){
\If{either $r_1=k$ and $S^{-1}R_1^{-1}R_2S$ is diagonal
      or $r_1=0$ and $r_2=k$}{
      \Return{$\Rvs\in \Svs$\;}
    }
  }
  \Begin(2.){
\lIf{either $r_1\leq \frac{\tilde{k}-1}2$ or $r_2\leq \frac{\tilde{k}-1}2$}{go to
     3.}\\
    \lElse{go to 4.}
  }
  \Begin(3. \fbox{Case $r_1\leq \frac{\tilde{k}-1}2$}
  \tcp*[f]{the case $r_2\leq \frac{\tilde{k}-1}2$ is analogous.})
  {
    \Return{$\rs\mat{0 & I}$\;}
  }  

  \Begin(4. \fbox{Case $\frac{\tilde{k}-1}2<r_2\leq r_1\leq
    \tilde{k}$} \tcp*[f]{the case $r_1\leq r_2$ is analogous.}){ 
    Find $M\in GL_k(\F_{q^k})$ such that $MS^{-1}(R_1 \ R_2)S$ is in
    row reduced echelon form\;
    $R(x):=MS^{-1}R_1S\Delta(x)-MS^{-1}R_2S$\; 
    Let $l_1,\dots,l_{\tilde{k}-r_1}\in \{1,\dots,k\}$ the columns of the
    pivots of $MS^{-1}R_2S$\; 
    Let $I':=(1,\dots,k)\setminus (l_1,\dots,l_{\tilde{k}-r_1},r_1+1,\dots,
    k)$\;
    Find $J',L'\subset I'$ satisfying Condition \eqref{e:cond_alg} and set $s:=|J|$\;  
    Let $K\subset I'\setminus (J'\cup L')$ with $|K|=\lfloor
    \frac{\tilde{k}+1}2\rfloor-\tilde{k}+r_1-s$\; 
    $\mu_i:=\left(\frac{\Min{J'\cup (i)}{L'\cup
          (i)}{R(0)}}{\Min{J'}{L'}{R(0)}}\right)^{[k-i]}$ for $i\in
    K$\; 
    \If{there exists an $i\in K$ such that $\rank(R(\mu_i))\leq
      \frac{\tilde{k}-1}2$}{ \Return $\rs\mat{I & S\Delta(\mu_i)S^{-1}}$\;
      \lElse{\Return there exists no $\Cvs\in \Svs$ such that
        $d(\Rvs,\Cvs)<k$;} } }
\end{algorithm}


\newpage

\subsection{A very efficient decoding algorithm 
for the case $R_1$ non singular}\label{ss:non_sing}

In this subsection, we focus on the case where the received word
$\Rvs=\rs \mat{R_1 & R_2}\in\Gr{\F_q}{k}{n}$ satisfies $R_1\in
GL_k(\F_q)$. In this case, we simplify the decoding algorithm and make
its complexity essentially negligible.

\medskip

We start by establishing the mathematical background.
Under the assumption that the matrix $R_1$ is invertible, an
alternative form of Theorem~\ref{t:unique_dec_sol} holds. 

\begin{proposition}\label{p:unique_non_sing}
  Let $\Rvs\in \Gr{\F_q}{k}{n}$ be a subspace with \[\frac{k-1}2 <
  \rank(R_2)\leq \rank(R_1)=k.\] The following are equivalent:
  \begin{itemize}
    \item $\Rvs$ is uniquely decodable.
  \item There exists a unique $\mu\in \F_{q^k}$ such that
    \begin{equation*}
      \rank(\Delta(\mu)-S^{-1}R_1^{-1}R_2S)=\mathrm{ndrank}
      (S^{-1}R_1^{-1}R_2S). 
    \end{equation*}
  \end{itemize}
\end{proposition}

\begin{proof} 
  By Theorem \ref{t:unique_dec_sol} $\Rvs$ is uniquely decodable if
  and only if there exists a unique $\mu\in \F_{q^k}$ such that
  \[\rank(\Delta(\mu)-S^{-1}R_1^{-1}R_2S)\leq\frac{k-1}2.\]
  Let $A=S\Delta(\mu)S^{-1}$, then by Corollary \ref{t:ndrank} 
  \[\rank(A-R_1^{-1}R_2)=\mathrm{ndrank}(\Delta(\mu)-S^{-1}R_1^{-1}R_2S)=
  \mathrm{ndrank}(S^{-1}R_1^{-1}R_2S).\]
\end{proof}

\medskip

Our improved decoding algorithm relies on the following corollary.

\begin{corollary}\label{c:non_sing}
  Let $\Rvs=\rs\mat{R_1& R_2}\in \Gr{\F_q}{k}{n}$ be uniquely
  decodable with $k=\rank(R_1)\geq \rank(R_2) >
  \frac{k-1}2$ and $S\in GL_k(\F_{q^k})$ a matrix
  dia\-go\-na\-li\-zing $P$. Let
  $R(x):=\Delta(x)-S^{-1}R_1^{-1}R_2S$. Then, for any choice of tuples
  of consecutive indices $J,L\subset (1,\dots,k)$ such that $J\cap
  L=\emptyset$ and $|J|=|L|=\mathrm{ndrank}(S^{-1}R_1^{-1}R_2S)$ it
  holds that for any $i\in (1,\dots,k)\setminus (J\cup L)$
  \[\rank\left(R\left(\left(
        \frac{\Min{J\cup (i)}{L\cup
            (i)}{S^{-1}R_1^{-1}R_2S}}{\Min{J}{L}{S^{-1}R_1^{-1}R_2S}}
      \right)^{[k-i]}\right)\right)\leq \frac{k-1}2.\] 

  Hence the unique $\mu\in\F_{q^k}$ from Proposition
  \ref{p:unique_non_sing} is
  \[\mu=\left( \frac{\Min{J\cup (i)}{L\cup
        (i)}{S^{-1}R_1^{-1}R_2S}}{\Min{J}{L}{S^{-1}R_1^{-1}R_2S}}
  \right)^{[k-i]}\]
  for any choice of $i\in (1,\dots,k)\setminus (J\cup L)$.
\end{corollary}

\begin{proof}
  By Proposition \ref{p:unique_non_sing}, there exists a unique $\mu$
  for which
  \[\rank(R(\mu))=\mathrm{ndrank}(S^{-1}R_1^{-1}R_2S) \leq
  \frac{k-1}2.\] 
  Hence it suffices to consider minors of $R(x)$ of size
  $\mathrm{ndrank}(S^{-1}R_1^{-1}R_2S)+1$.

  By Corollary \ref{t:ndrank}, the minor 
  \[\Min{J\cup (i)}{L\cup (i)}{R(x)}=\Min{J}{L}{S^{-1}R_1^{-1}R_2S}
  x^{[i]}-\Min{J\cup (i)}{L\cup (i)}{S^{-1}R_1^{-1}R_2S}\]
  is not identically zero. Hence the root 
  \[\mu=\left( \frac{\Min{J\cup (i)}{L\cup
        (i)}{S^{-1}R_1^{-1}R_2S}}{\Min{J}{L}{S^{-1}R_1^{-1}R_2S}}
  \right)^{[k-i]}\]
  makes $\rank(R(\mu))=\mathrm{ndrank}(S^{-1}R_1^{-1}R_2S)$. By
  Proposition~\ref{p:unique_non_sing}, $\mu$ yields the unique solution
  to the decoding problem.
\end{proof}

\begin{remark}
  The previous corollary allows us to design a more efficient decoding
  algorithm than the one presented in~\cite{ma08p}, since it does not
  require the use of the Euclidean Algorithm. More precisely, it allows us to
  find a minor (in fact, many of them) whose roots can be directly
  computed via an explicit formula. Practically, this makes the
  decoding complexity negligible. 
\end{remark}

\begin{algorithm}\label{a:non_sing}
  \caption{Minimum-distance decoding algorithm: $n=2k$, $R_1$ non-singular}
  \Input{$\Rvs=\rs\mat{R_1 & R_2}\in \Gr{\F_q}{k}{2k}$ with either
    $\rank(R_1)=k$ or $\rank(R_2)=k$ \\
    $P\in GL_k(\F_q)$ the companion matrix of $p\in \F_q[x]$ and \\
    $S\in GL_k(\F_{q^k})$ its diagonalizing matrix.} 
  \Output{$\Cvs \in\Svs \subset \Gr{\F_q}{k}{n}$ such that
    $d(\Rvs,\Cvs)<k$, if such a $\Cvs$ exists.}
  \BlankLine
  Let $r_i:=\rank(R_i)$ for $i=1,2$.\\
  \Begin(1.){
    \If{either $r_1=k$ and $S^{-1}R_1^{-1}R_2S$ is diagonal
      or $r_1=0$ and $r_2=k$}{
      \Return{$\Rvs\in \Svs$\;}
    }
  }
  \Begin(2.)
  {
    \lIf{either $r_1\leq \frac{k-1}2$ or $r_2\leq \frac{k-1}2$}{go to
     3.}\\
    \lElse{go to 4.}
  }
  \Begin(3. \fbox{Case $r_1\leq \frac{k-1}2$}
  \tcp*[f]{the case $r_2\leq \frac{k-1}2$ is analogous.})
  {
    \Return{$\rs\mat{0 & I}$\;}
  }  
\Begin(4. \fbox{Case $r_1=k$} \tcp*[f]{the case 
    $r_2=k$ is analogous.}){
    $R(x):=\Delta(x)-S^{-1}R_1^{-1}R_2S$\;
    $s:=\rank\subm{(1,\dots,\lfloor\frac{k-1}2\rfloor)}{(k-\lfloor
      \frac{k-1}2\rfloor+1,\dots,k)}{R(0)}$\;
    $\mu:=\frac{\Min{(1,2,\dots,s+1)}{(1,k-s,\dots,k)}{R(0)}}{\Min{(2,\dots,
        s+1)
      }{(k-s,\dots,k)}{R(0)}}$\;
    \If{$\rank\left(R\left(\mu\right)\right)\leq \frac{k-1}2$}{
      \Return{$\rs\mat{I & S\Delta(\mu)S^{-1}}\in
        \Svs$\;} \lElse{\Return there exists no $\Cvs\in \Svs$ such that
        $d(\Rvs,\Cvs)<k$;}}
  }

\end{algorithm}

\section{Algorithms Complexities} \label{s:algorithm}

In this section, we compute the complexity of some algorithms that we
gave in the previous section. 

We start by specifying an algorithm for finding tuples $J',L'\subset I'$
needed in Step 4 of Algorithm~\ref{a:dec_alg}. The algorithm performs
only row operations. The pseudocode is given in
Algorithm~\ref{a:alg_min}, while correctness is proved in the next
lemma.

\begin{lemma}
Let $M\in \F_q^{k\times k}$ be a non diagonal
matrix. Algorithm~\ref{a:alg_min} finds two tuples $J,L\subset
(1,\dots,k)$ such that $J,L\neq \emptyset$, $J\cap L=\emptyset$,
$\Min{J}{L}{}\neq 0$ and $\Min{J\cup (j)}{L\cup (l)}{}=0$ for any
$j,l\in (1,\dots,k)\setminus (J\cup L)$, $j\neq l$.
\end{lemma}

\begin{proof} We start by setting $K=(1,\ldots,k)$.
  The algorithm eventually terminates since $|K|$ strictly decreases
  after every cycle of the while loop. Moreover, its complexity is
  bounded by the complexity of the Gaussian elimination algorithm
  which computes the row reduced echelon form of a matrix of
  $\F_q^{n\times n}$ in $\mathcal{O}(\F_q;n^3)$ operations.

  We have to prove that the returned tuples $J,L\subset (1,\dots,k)$
  satisfy the output conditions. Since $M$ is not diagonal,
  $J,L\neq\emptyset$. The emptiness of $J\cap L$ follows from the fact
  that $J,L$ are initialized to $\emptyset$ and each time we modify
  them, we get $J\cup (j)$ and $L\cup (l)$ where $j\neq l$ and $j,l$ are not
  elements of $J\cup L$.

  In order to continue we have to characterize the matrix $N$. The
  matrix changes as soon as we find coordinates $j,l\in I$ with $i\neq
  j$ for which $n_{jl}\neq 0$. The multiplication $PN$ consists of the
  following row operations
  \begin{itemize}
  \item the $i$-th row of $PN$ is the $i$-th row of $N$
    for $i\leq j$, and
  \item the $i$-th row of $PN$ is the $i$-th row of $N$ minus
    $\frac{n_{i,l}}{n_{j,l}}$ times the $j$-th row of $N$, where
    $N=(n_{j,l})_{\substack{1\leq j,l\leq k}}$ for $i>j$.
  \end{itemize}
  It follows that the entries of the $l$-th column of $PN$ are zero as
  soon as the row index is bigger than $j$.

  We claim that after each cycle of the while loop it holds that
  $\Min{J}{L}{N}\neq 0$. We prove it by induction on the cardinality
  of $J$ and $L$. Since the matrix $M$ is not diagonal, the while loop
  will eventually produce tuples $J=(j)$ and $L=(l)$ with $j\neq l$
  such that $\Min{J}{L}{M}\neq 0$. Now suppose that we have $J,L$ such
  that $J,L\neq\emptyset$, $J\cap L=\emptyset$ and $\Min{J}{L}{N}\neq
  0$ and there exist, following the algorithm, entries $j,l\in I$ with
  $j\neq l$ such that $n_{j,l}\neq 0$. From the previous paragraph,
  the only nonzero entry of the row with index $j$ of $\subm{J\cup
    (j)}{L\cup (l)}{N}$, which by construction is the last one, is
  $n_{j,l}$, hence
  \[\Min{J\cup (j)}{L\cup (l)}{N}=n_{j,l}\Min{J}{L}{N}\neq 0.\]
  In order to conclude that $\Min{J}{L}{M}\neq 0$, it is enough to
  notice that the row operations bringing $\subm{J}{M}{M}$ to
  $\subm{J}{M}{N}$ are rank preserving.

  The property of maximality of the minor $\Min{J}{L}{M}$ with respect
  to containment is a direct consequence of the structure of the algorithm.
\end{proof}

\begin{algorithm}\label{a:alg_min}
  \caption{Modified Gaussian elimination}
  \Input{$M\in \F_q^{k\times k}$ non diagonal matrix.} 
  \Output{$J,L\subset (1,\dots,k)$ such that $J,L\neq \emptyset$,
    $J\cap L=\emptyset$, $\Min{J}{L}{}\neq 0$ and $\Min{J\cup
      (j)}{L\cup (l)}{}=0$ for any $j\neq l\in (1,\dots,k)\setminus (J\cup
    L)$.}
  \BlankLine
  $J=L=\emptyset$, $K=(1,\dots,k)$, $j=1$ and $N=(n_{j,l})_{\substack{1\leq
      j,l\leq k}}=M$\;
  \While{$K\neq \emptyset$}
  {
    $t:=0$\;
    \For{$l\in K$ and $l\neq j$}{
      \If{$n_{j,l}\neq 0$ and $t=0$}
      {
        $J=J\cup (j), L=L\cup (l)$ and $K=K\setminus(j,l)$\;
        $P=(p_{j',l'})_{\substack{1\leq j',l'\leq k}}$ such that $p_{i,i}=1$
         for any $i\in \{1,\dots,k\}$, $p_{i,l}=-\frac{n_{i,l}}{n_{j,l}}$
        for any $i\in I$ with $i>j$ and $p_{j',l'}=0$ otherwise\;
        $N=PN$\;
        $t=1$\;
      }
    }
    \lIf{$t=0$}{$K=K\setminus (j)$\;}
     $j=\min K$\;
  }
  \Return{$J,L$\;}
\end{algorithm}

For simplicity, in the following comparisons we give the
minimum-distance decoding complexity only for the case when the
received space $\Rvs\in \Gr{\F_q}{k}{n}$. This is an upper bound for
the complexity in the general case. The precise complexity for the
case when $\Rvs\in \Gr{\F_q}{\tilde{k}}{n}$, $\tilde{k}<k$ may be obtained
via an easy adaptation of our arguments.

\subsubsection*{Complexity of the decoding algorithm}\label{sss:complexity}

Algorithm \ref{a:dec_alg} consists of matrix operations over the
extension field $\F_{q^k}\supseteq\F_q$. The most expensive of such
operations is the computation of the rank of matrices of size $k\times
k$, which can be performed via the Gaussian elimination algorithm. The
complexities then are as follows:
\begin{itemize}
\item The complexity of step 4. is $\mathcal{O}(\F_{q^k};k^3)$, which corresponds to the
  computation of $\rank(R(\mu))$.
\item The complexity of step 5. is $\mathcal{O}(\F_{q^k};k^4)$, which corresponds to the
  computation of $\rank(R(\mu_i))$ for all $i\in K$, where $|K|\leq
  \lfloor\frac{k-1}2\rfloor$.
\end{itemize}

The overall complexity of Algorithm \ref{a:dec_alg} is then
$\mathcal{O}(\F_{q^k};k^4)$. This makes the complexity of Algorithm
\ref{a:dec_alg_n=rk} $\mathcal{O}(\F_{q^k};(n-k)k^3)$. Notice that
computing the rank of the matrices $R_i$ has complexity
$\mathcal{O}(\F_q;(n-k)k^2)$, which is dominated by $\mathcal{O}(\F_{q^k};(n-k)k^3)$. 

\subsubsection*{Comparison with other algorithms and conclusions}

We compare the complexity of Algorithm \ref{a:dec_alg_n=rk} with other
algorithms present in the literature, specifically with the
algorithms discussed in Proposition~\ref{p:compl}. The complexity of
the decoding algorithm contained in \cite{ko08} is
$\mathcal{O}(\F_{q^{n-k}};n^2)$. In order to compare the two complexity
estimates, we use the fact that the complexity of the operations on an
extension field $\F_{q^s}\supseteq\F_q$ is $\mathcal{O}(\F_q;s^2)$. This is a
crude upper bound, and the complexity may be improved in some cases
(see, e.g., \cite{go07u2}). Nevertheless, under this assumption the
decoding algorithm from~\cite{ko08} has complexity $\mathcal{O}(\F_q;n^2(n-k)^2)$.

Following similar reasoning, the complexity of the decoding algorithm
contained in~\cite{si08a} is $\mathcal{O}(\F_q^{n-k};k(n-k))$, i.e.,
$\mathcal{O}(\F_q;k(n-k)^3)$.

We conclude that the minimum-distance decoding algorithm presented in
this paper has lower complexity than the algorithms in~\cite{ko08}
and \cite{si08a}, whenever $k\ll n$. Since this is the relevant
case for the applications, the decoding algorithm
that we propose constitutes usually a faster option for decoding
spread codes. 

\section{Conclusions}

In this paper we exhibit a minimum distance decoding algorithm for
spread codes which performs better than other known decoding
algorithms for RSL codes when the dimension of the codewords is small
with respect to the dimension of the ambient space.

The problem of extending our decoding algorithm to the case when the
dimension of the received space is bigger than the dimension of the
codewords remains open. Another natural question arising from this
work is finding a generalization of the decoding algorithm to a
list decoding algorithm. Theorem~\ref{t:unique_dec_sol} can be easily
extended for this purpose. Yet finding a way to solve the list
decoding problem which requires neither the computation of a
gcd, nor the factorization of a minor is a non trivial task.

\section*{Acknowledgement}
The authors are grateful to Heide Gluesing-Luerssen for useful
discussions and help in the proof of Theorem~\ref{t:reg_mat}.

\def\cprime{$'$} \def\polhk#1{\setbox0=\hbox{#1}{\ooalign{\hidewidth
  \lower1.5ex\hbox{`}\hidewidth\crcr\unhbox0}}}
  \def\polhk#1{\setbox0=\hbox{#1}{\ooalign{\hidewidth
  \lower1.5ex\hbox{`}\hidewidth\crcr\unhbox0}}} \def\cprime{$'$}
  \def\cprime{$'$} \def\cprime{$'$} \def\cprime{$'$}

\end{document}